\documentclass[article,pra,onecolumn,superscriptaddress,nofootinbib,longbibliography,floatfix]{revtex4-2}[a4paper]

\usepackage{lmodern}
\usepackage[utf8]{inputenc}
\usepackage{amsmath,amssymb,amsthm}
\usepackage{physics}
\usepackage{graphicx}
\usepackage{booktabs}
\usepackage{microtype}
\usepackage{tikz}
\usetikzlibrary{arrows.meta,shapes.geometric,positioning}
\usepackage{quantikz}
\usepackage[caption=false]{subfig} 
\usepackage{algorithm}
\usepackage[noend]{algpseudocode}
\usepackage{enumitem}
\usepackage[colorlinks=true, allcolors=blue]{hyperref}

\newcommand{\maybeincludegraphics}[2]{%
	\IfFileExists{#1}{\includegraphics[width=#2]{#1}}{\fbox{\parbox[c][0.2\textheight][c]{#2}{\centering \texttt{#1} not found}}}%
}

\newtheorem{theorem}{Theorem}[section]
\newtheorem{lemma}[theorem]{Lemma}

\theoremstyle{definition}

\usepackage{hyperref}
\hypersetup{
	colorlinks=true,
	linkcolor=blue,
	citecolor=blue,
	urlcolor=blue,
	pdftitle={Learning a Nonlinear Entanglement Witness with CV-QNNs},
	pdfauthor={M. Rezaei Shokouh and H. Davoodi Yeganeh}
}

\begin{document}
	
	\title{Hybrid Quantum-Classical Learning of Nonlinear Entanglement Witnesses via Continuous-Variable Quantum Neural Networks}
	
	\author{Mohammad Rezaei Shokouh}
	\affiliation{Department of Physics, Ferdowsi University of Mashhad, Mashhad, Iran} 
	
	\author{Hossein Davoodi Yeganeh \footnote {corresponding author: h.yeganeh@ariaquanta.com, h.yeganeh@ssau.ac.ir}}
	\affiliation{AriaQuanta Quantum Computing Center, Tehran, Iran}
	\affiliation{Quantum Research Center, Shahid Sattari University of Aeronautical Sciences and Technology, Tehran, Iran}

	

	\begin{abstract}
A major challenge in quantum information is characterizing entanglement, for which entanglement witnesses offer effective means of detecting quantum correlations. We introduce a hybrid quantum-classical framework that learns a nonlinear entanglement witness directly from quantum data using continuous-variable quantum neural networks (CV-QNNs). Our architecture combines variational interferometers, squeezers and non-Gaussian Kerr gates with a small classical neural head to output a scalar witness value. Numerical simulations were conducted on two- and three-mode families, including Gaussian and non-Gaussian states in both pure and mixed forms. We observed over $99\%$ classification accuracy and a robust performance gap compared to strong classical baselines, especially when scaling from two to three modes. Robustness to photon loss is further quantified under a finite number of measurement shots. On the theory side, we show that when the quantum measurement stage is informationally complete, the hybrid model can approximate any continuous witness-like functional on compact sets of states.Our findings highlight CV-QNNs as a promising framework for data-driven quantum state characterization and propose specific benchmarks where near-term photonic platforms offer tangible advantages.
	\\

	\end{abstract}

	\maketitle
	
	\section{Introduction}
	\label{sec:introduction}
	
	Entanglement is a cornerstone of the second quantum revolution, serving as a fundamental resource that underpins potential advances in quantum computing, secure communication, and precision metrology \cite{r1, r2}. The ability to generate, manipulate, and, crucially, certify the presence of this uniquely quantum correlation is therefore a paramount task for the advancement of quantum technologies. However, the problem of distinguishing entangled states from separable ones—the separability problem—is known to be NP-hard in the general case \cite{r3}, presenting a formidable challenge as quantum systems scale in size and complexity.
	
	Historically, analytical and operational tools have been developed to tackle this issue. For low-dimensional bipartite systems, criteria like the Positive Partial Transpose (PPT) condition provide a powerful, necessary, and sufficient condition for separability \cite{r4}. Yet, its sufficiency fails in higher dimensions and for multipartite systems, where the geometry of the set of separable states becomes notoriously complex. The primary tool for experimental certification has been the entanglement witness (EW) formalism \cite{r5, r6}. A linear EW is a Hermitian operator $W$ constructed such that its expectation value is non-negative for all separable states ($\mathrm{Tr}(W\rho) \ge 0$) but negative for at least one entangled state. Geometrically, it defines a separating hyperplane in the space of density matrices. While powerful, this approach suffers from two major limitations: (i) a single linear witness is inherently non-universal, capable of detecting only a subset of entangled states, and (ii) finding an optimal witness for an arbitrary, unknown state is a computationally hard problem. While nonlinear witnesses can define more complex decision boundaries, their analytical construction is generally intractable.
	
	This intractability motivates a paradigm shift toward data-driven approaches, where a decision function—a witness—is not derived from first principles but is learned from examples \cite{r7,r22,r23}. Initial efforts have successfully applied classical machine learning (ML) models, such as support vector machines and neural networks, to this task \cite{r8, r9,r26,r27}. However, existing models typically rely on classical representations of quantum states, such as the elements of a tomographically reconstructed density matrix. This methodology introduces a fundamental bottleneck: it depends on a classical shadow of the quantum state—a process that is often resource-intensive and may obscure the subtle, high-dimensional correlations indicative of genuine multipartite entanglement. Moreover, as the Hilbert space dimension increases exponentially, classical machine learning models inevitably encounter the curse of dimensionality, making it difficult to learn the complex boundary of the separable set from a feasible number of examples.
	
	A more native and potentially more powerful approach is to utilize Quantum Machine Learning (QML) models that can process quantum data directly \cite{r10}. Among the various platforms for QML, Continuous-Variable (CV) quantum computing, particularly in photonic systems, offers a natural framework for encoding and manipulating quantum information \cite{r11, r12,r24}. The layered, variational architecture of a Continuous-Variable Quantum Neural Network (CV-QNN), as introduced by Killoran et al. \cite{r13}, provides a compelling ansatz for learning complex functions of quantum states. Inspired by classical deep learning, a CV-QNN consists of a sequence of layers, each performing a Gaussian transformation (analogous to a classical affine map) followed by a non-Gaussian operation (analogous to a nonlinear activation function), granting it universal function approximation capabilities within the quantum domain.
	
	In this work, we propose and numerically validate a CV-QNN designed to function as a universal, trainable, and nonlinear entanglement witness. Our central hypothesis is that by processing quantum states directly within their native Hilbert space, such a model can overcome the limitations of classical ML and exhibit a performance advantage that scales with system complexity. We test this hypothesis by training our hybrid quantum-classical model to distinguish separable from entangled states in both two-mode and three-mode scenarios, using a rich dataset of Gaussian and non-Gaussian states. Our main contributions are as follows.
	 We introduce the theoretical foundations, from a brief overview of continuous-variable quantum mechanics to the formulation of  CV-QNN in Section \ref{sec:theory}.  The next section \ref{sec:architect} presents an overview of our hybrid quantum-classical pipeline, covering architectural details, variational ansatze, data generation and preprocessing steps, training procedures, and the evaluation protocol used for comparison with classical models. In section \ref{sec:results}, we report numerical experiments assessing performance, scalability, and robustness of the proposed model, alongside a theoretical justification and discussion analysis. The final section \ref{sec:con} presents our conclusions and summarizing the key findings.

	\section{Theoretical Background}
	\label{sec:theory}
	
	In this section, we provide the necessary theoretical foundations, beginning with a brief review of the continuous-variable formalism of quantum mechanics and culminating in the precise formulation of our CV-QNN as a learnable, nonlinear witness functional.
	
	\subsection{Continuous-Variable Systems and Phase-Space Formalism}
	
	While quantum information is often introduced in the context of discrete, two-level systems (qubits), many physical systems, particularly those in quantum optics, are naturally described by continuous degrees of freedom. These are known as Continuous-Variable (CV) systems \cite{r12, r13}. The quantum state of an $M$-mode CV system lives in an infinite-dimensional Hilbert space $\mathcal{H} = \bigotimes_{i=1}^M \mathcal{H}_i$, where each $\mathcal{H}_i$ is a Fock space spanned by the number states $\{\ket{n}\}_{n=0}^\infty$.
	
	The state of each mode is defined by its creation and annihilation operators, $\hat{a}^\dagger$ and $\hat{a}$, which obey the canonical commutation relation $[\hat{a}, \hat{a}^\dagger]=1$. An equivalent and often more intuitive description is provided by the dimensionless quadrature operators, analogous to classical position and momentum:
	\begin{equation}
		\hat{x} = \frac{\hat{a} + \hat{a}^\dagger}{\sqrt{2}}, \quad \hat{p} = \frac{\hat{a} - \hat{a}^\dagger}{i\sqrt{2}}.
	\end{equation}
	These operators are Hermitian and satisfy the commutation relation $[\hat{x}, \hat{p}] = i$, which directly leads to the Heisenberg uncertainty principle, $\Delta x \Delta p \ge 1/2$. This formalism allows us to represent quantum states in a 2D phase space spanned by the eigenvalues of $\hat{x}$ and $\hat{p}$.
	
	A powerful tool for this representation is the Wigner function, $W(x, p)$, a quasi-probability distribution that provides a complete description of the quantum state. While it can take negative values for non-classical states—a key signature of quantumness—its marginals correspond to the probability distributions of the quadratures. States whose Wigner function is a positive-definite Gaussian function are known as Gaussian states. These include the vacuum state $\ket{0}$, coherent states $\ket{\alpha}$, and squeezed states. Geometrically, their Wigner function is represented by an ellipse in phase space, whose center, orientation, and axis lengths are determined by the first and second moments of the quadrature operators.
	
	Unitary operations in CV systems are generated by Hamiltonians that are polynomials in the mode operators. Operations generated by Hamiltonians at most quadratic in these operators are called Gaussian operations, as they preserve the Gaussian character of states. The fundamental set of single- and two-mode Gaussian gates includes:
	\begin{itemize}[leftmargin=*]
		\item \textbf{Displacement}, $D(\alpha) = \exp(\alpha \hat{a}^\dagger - \alpha^* \hat{a})$: This operation shifts the center of the state in phase space by a complex amount $\alpha$.
		\item \textbf{Rotation (Phase)}, $R(\phi) = \exp(i\phi \hat{n})$: This rotates the state around the origin in phase space by an angle $\phi$.
		\item \textbf{Squeezing}, $S(z) = \exp(\frac{1}{2}(z^* \hat{a}^2 - z \hat{a}^{\dagger 2}))$ where $z=re^{i\phi}$: This operation attenuates the variance of one quadrature (squeezing) at the expense of amplifying the variance of the orthogonal quadrature. In phase space, this corresponds to deforming a circular vacuum state into an ellipse.
		\item \textbf{Beamsplitter}, $BS(\theta, \phi) = \exp(\theta(e^{i\phi}\hat{a}_i^\dagger \hat{a}_j - e^{-i\phi}\hat{a}_i \hat{a}_j^\dagger))$: This is a two-mode gate that interferes or mixes two modes.
	\end{itemize}
	Although Gaussian operations serve as essential building blocks in  CV quantum computing, they are insufficient for achieving universality. According to the Gottesman-Knill theorem in the discrete-variable regime, computations restricted to Gaussian states, operations, and measurements can be efficiently simulated on classical hardware. To harness the full potential of CV quantum computation and access genuinely non-classical features, the inclusion of at least one non-Gaussian component is necessary \cite{r14}.
	
	A canonical choice for a non-Gaussian gate is one generated by a Hamiltonian of third order or higher in the mode operators. In this work, we utilize the Kerr gate, generated by the Hamiltonian $H_{\text{Kerr}} \propto (\hat{a}^\dagger \hat{a})^2 = \hat{n}^2$. The corresponding unitary is given by:
	\begin{equation}
		K(\kappa) = \exp(i\kappa \hat{n}^2).
	\end{equation}
	This gate imparts a phase shift that is proportional to the square of the photon number, a distinctly nonlinear effect. In phase space, it distorts the shape of the Wigner function, capable of transforming a Gaussian state into a complex, non-Gaussian one. This ability to introduce non-Gaussianity is essential for the CV-QNN to learn highly non-trivial functions, making the Kerr gate the direct analogue of a nonlinear activation function in a classical neural network.
	
	\subsection{Continuous-Variable Quantum Neural Network Layers}
	
	The architecture of a Continuous-Variable Quantum Neural Network (CV-QNN) is designed to mirror the structure of a classical feedforward neural network, which consists of a series of interconnected layers. In a classical network, each layer typically performs an affine transformation on its input vector (a linear map followed by a bias shift), and then applies a nonlinear activation function. The CV-QNN layer accomplishes an analogous transformation, not on a classical vector, but directly on a quantum state living in the phase space of the system \cite{r13}.
	
	A single layer of our CV-QNN, denoted by $L_k$ with a set of trainable parameters $\theta_k$, is a composite unitary transformation built from four distinct physical operations. The sequence of these operations is critical and follows the order implemented in our numerical simulations: an interferometer, a squeezing layer, a displacement layer, and finally, a nonlinear Kerr layer. The full transformation for a single layer acting on an $M$-mode state is thus:
	\begin{equation}
		L_k(\theta_k) = K(\vec{\kappa}_k) \circ D(\vec{\alpha}_k) \circ S(\vec{z}_k) \circ U_{\text{I}}(\vec{\phi}_k, \vec{\varphi}_k).
	\end{equation}
	
	\subsubsection{The Interferometer: A Trainable Linear Transformation}
	The first stage of the layer is a multi-mode passive optical interferometer, $U_{\text{I}}$. This operation performs a general energy-conserving linear transformation on the mode operators, effectively mixing the information contained in the different modes. An arbitrary $M$-mode linear interferometer can be decomposed into a mesh of two-mode beamsplitters ($BS$) and single-mode phase shifters ($R$) \cite{r28}. For our $M$-mode system, the unitary is given by:
	\begin{equation}
		U_{\text{I}}(\vec{\phi}, \vec{\varphi}) = \left( \prod_{i<j} BS_{ij}(\phi_{ij}) \right) \left( \prod_{i=1}^M R_i(\varphi_i) \right).
	\end{equation}
	This operation is analogous to the multiplication by a weight matrix $\mathbf{W}$ in a classical neural network. The parameters of the phase shifters and beamsplitters, $\vec{\phi}$ and $\vec{\varphi}$, are trainable and correspond to the elements of this matrix. In phase space, this unitary transformation corresponds to a rotation of the entire state's Wigner function representation.
	
	\subsubsection{Gaussian Affine Transformations: Squeezing and Displacement}
	Following the interferometer, two layers of single-mode Gaussian gates are applied to perform an affine transformation on each mode individually.
	
	\textbf{The Squeezing Layer}, $S(\vec{z}) = \bigotimes_{i=1}^M S_i(z_i)$, applies a squeezing operation to each mode. The squeezing gate $S(z_i)$ transforms the quadrature operators as $\hat{x}_i \to e^{-r_i}\hat{x}_i$ and $\hat{p}_i \to e^{r_i}\hat{p}_i$ (for $z_i=r_i$). This deforms the Wigner function of each mode, attenuating the variance in one quadrature while amplifying it in the orthogonal one. This is analogous to a diagonal scaling operation in a classical network, allowing the model to adjust the importance or scale of features in each mode's phase space.
	
	\textbf{The Displacement Layer}, $D(\vec{\alpha}) = \bigotimes_{i=1}^M D_i(\alpha_i)$, applies a displacement to each mode. The operator $D(\alpha_i)$ shifts the quadratures by a constant amount, $\hat{x}_i \to \hat{x}_i + \sqrt{2}\text{Re}(\alpha_i)$ and $\hat{p}_i \to \hat{p}_i + \sqrt{2}\text{Im}(\alpha_i)$. This is the direct analogue of adding a bias vector $\mathbf{b}$ in a classical layer.
	
	Together, the interferometer, squeezing, and displacement layers collectively implement the full affine transformation $\vec{q} \mapsto \mathbf{W'}\vec{q} + \vec{b'}$ on the vector of quadrature operators $\vec{q} = (\hat{x}_1, \hat{p}_1, \dots, \hat{x}_M, \hat{p}_M)^T$. At this point, the transformation is still entirely Gaussian.
	
	\subsubsection{The Kerr Gate: A Nonlinear Activation Function}
	The final and most crucial component of the layer is the application of a non-Gaussian gate, which introduces nonlinearity and grants the network its computational power. We use the single-mode Kerr gate, $K(\kappa) = \exp(i\kappa \hat{n}^2)$, applied to each mode:
	\begin{equation}
		K(\vec{\kappa}) = \bigotimes_{i=1}^M K_i(\kappa_i).
	\end{equation}
	The Hamiltonian $\hat{n}^2$ is of fourth order in the mode operators, making this a non-Gaussian operation. Its effect is to introduce a phase shift that depends on the square of the photon number, a powerful nonlinear effect. This is precisely analogous to a nonlinear activation function $\sigma(\cdot)$ in a classical neural network. It deforms the Wigner function in a complex, non-elliptical way, enabling the model to learn decision boundaries that are inaccessible to purely Gaussian models. The strength of this nonlinearity, $\kappa_i$, is a trainable parameter for each mode.
	
	\subsubsection{Stacking Layers to Form a Deep Network}
	A complete CV-QNN is constructed by composing $L$ such layers sequentially. The full variational unitary is thus given by:
	\begin{equation}
		U(\Theta) = L_L(\theta_L) \circ L_{L-1}(\theta_{L-1}) \circ \cdots \circ L_1(\theta_1),
	\end{equation}
	where $\Theta = (\theta_1, \dots, \theta_L)$ represents the entire set of trainable quantum parameters. This deep, layered structure allows the network to learn progressively more abstract and complex features of the input quantum state, building a highly expressive variational ansatz for approximating the desired entanglement witness function.
	
	\subsection{Learned Nonlinear Witness via Informationally Complete Readout}
	
	The standard formalism of an entanglement witness relies on a Hermitian operator. However, our data-driven model, which uses a classical neural network to process measurement outcomes, does not inherently learn such an operator. Instead, it learns a more general object: a nonlinear witness functional, $W_{\Theta,\Phi}(\rho)$, which maps a density matrix $\rho$ to a real number. The model is trained such that this functional respects the geometric separation of states: positive for all separable states in the training set and negative for the entangled ones. The decision boundary, $W_{\Theta,\Phi}(\rho)=0$, is thus a highly flexible, nonlinear surface optimized to approximate the true separability boundary.
	
	A critical component in defining this functional is the measurement process. A simple projective measurement in a single, fixed basis (such as the Fock basis after the variational circuit) is generally not informationally complete (IC). This means that two distinct quantum states, $\rho_1 \neq \rho_2$, could yield the exact same probability distribution, $p(n) = \langle n|\rho |n \rangle$, making them indistinguishable to the classical network that follows. This degeneracy would fundamentally limit the model's ability to act as a universal witness, as it would be blind to quantum information stored in the coherences (off-diagonal elements) of the density matrix in that basis.
	
	To overcome this limitation, the feature vector fed to the classical head must be extracted via an IC Positive Operator-Valued Measure (POVM), denoted $\{\mathcal{M}_k\}$. An IC-POVM guarantees that the mapping from a state to its vector of outcome probabilities, $\rho \mapsto \vec{p}(\rho) = (\mathrm{Tr}(\mathcal{M}_1\rho), \mathrm{Tr}(\mathcal{M}_2\rho), \dots)^T$, is injective (one-to-one) \cite{r15}. In the context of CV systems, this can be realized in several ways:
	\begin{enumerate}[leftmargin=*]
		\item \textbf{Heterodyne or Homodyne Tomography:} These standard techniques in quantum optics measure quadratures in rotated bases and are known to be informationally complete \cite{r16}.
		\item \textbf{Ensemble of Measurement Bases:} A more practical approach for near-term devices, and the one we implicitly model, is to measure in a fixed basis (e.g., Fock basis) after applying a set of different, known pre-measurement unitaries, $\{V_1, V_2, \dots, V_K\}$.
	\end{enumerate}
	The collection of probabilities from this set of measurements forms an informationally complete feature vector. Our witness functional is therefore more precisely defined as:
	\begin{equation}
		W_{\Theta,\Phi}(\rho) = f_\Phi\Big( \mathrm{concat}\big[ \vec{p}_{\Theta,V_1}(\rho), \dots, \vec{p}_{\Theta,V_K}(\rho) \big] \Big),
	\end{equation}
	where $\vec{p}_{\Theta,V_k}(\rho)$ is the vector of Fock probabilities after the state is transformed by $V_k \circ U_\Theta$. This construction ensures that the classical head receives a faithful representation of the quantum state, containing all the information necessary for classification.
	
	This rigorous foundation allows us to state the expressive power of our model in a formal theorem, which is a cornerstone of our theoretical claim.
	
	\begin{theorem}[Approximation Capability of the Learned Witness]
		Let $\mathcal{X}$ be a compact set of density operators on a finite-dimensional Hilbert space $\mathcal{H}_d$. Let the measurement scheme, defined by a variational unitary $U_\Theta$ and an informationally complete POVM, produce a continuous and injective feature map $g_\Theta: \rho \mapsto \vec{p}(\rho) \in \mathbb{R}^K$. Then the class of witness functionals
		\[\mathcal{F} = \{ \rho \mapsto f_\Phi(g_\Theta(\rho)) \mid \Theta \in \mathcal{P}_\Theta, \Phi \in \mathcal{P}_\Phi \}\]
		is dense in the space of all continuous real-valued functions on $\mathcal{X}$, $C(\mathcal{X})$, under the uniform norm.
	\end{theorem}
	
	The proof of this theorem, detailed in Appendix \ref{app:universal_witness_proof}, rests on two pillars. First, the injectivity of the IC measurement map ensures that any two different quantum states are mapped to two different classical vectors (point separation). Second, the Universal Approximation Theorem for classical neural networks guarantees that an MLP, $f_\Phi$, can approximate any continuous function on this space of classical vectors. Combining these two properties confirms that our hybrid architecture is, in principle, capable of learning and approximating any continuous entanglement witness functional to arbitrary precision. It is important to note that this is a statement about the model's expressivity, not its trainability or generalization from finite data, which we investigate numerically.
	
	\section{Architecture and Methodology}
	\label{sec:architect}
	
	In this section, we provide a comprehensive, reproducible description of our hybrid quantum-classical pipeline. We detail the model architecture, the specific variational ansatze used in our experiments, the data generation and preprocessing pipeline, the training protocol, and the methodology for ensuring a fair comparison against strong classical baselines.
	
	\subsection{Hybrid Quantum-Classical Model Overview}
	
	Our model is a hybrid classifier designed to learn a nonlinear witness functional $W_{\Theta,\Phi}(\rho)$ directly from data. As depicted schematically in Fig.~\ref{fig:schematic}, the pipeline consists of four main stages:
	
	\begin{enumerate}[leftmargin=*]
		\item \textbf{Quantum Feature Map:} An input mixed state $\rho_{\text{in}}$ is mapped onto a larger, purified Hilbert space. In contrast to a true purification (which is not physically realizable for an unknown state), this is achieved by a physically grounded feature map: tensoring the input system with ancillary modes prepared in a known state (e.g., vacuum $\ket{0}$) and applying a fixed, non-trainable unitary. This prepares the input for the variational circuit.
		
		\item \textbf{Variational Quantum Circuit (CV-QNN):} The prepared state is then processed by the core of our model, a deep, parameterized CV-QNN circuit, $U(\Theta)$. This circuit, detailed in the next subsection, learns to transform the state into a representation where separability and entanglement are more easily distinguished.
		
		\item \textbf{Informationally Complete (IC) Measurement:} To extract a faithful classical representation of the transformed quantum state, we perform an IC measurement. In our simulations, this is implemented by applying a set of $K$ fixed, pre-measurement unitaries $\{V_k\}$ and then measuring the photon number probabilities in the Fock basis for each. This generates a rich, concatenated feature vector $\vec{p} \in \mathbb{R}^{K \times d^M}$.
		
		\item \textbf{Classical Head:} The resulting classical feature vector is fed into a small, classical feedforward neural network (MLP), $f_\Phi$, which processes these features and outputs a single scalar value, interpreted as the learned witness $W_{\Theta,\Phi}(\rho)$. The final classification is determined by the sign of this value.
	\end{enumerate}
	
	This four-stage process, visualized in Fig.~\ref{fig:classification_flow}, allows for an end-to-end differentiable model where both the quantum and classical parameters can be co-optimized.
	
	\tikzset{
		box/.style={draw, rectangle, rounded corners, thick, text centered, minimum height=1.2cm, minimum width=2.6cm},
		qbox/.style={box, fill=blue!10},
		cbox/.style={box, fill=green!10},
		data/.style={trapezium, trapezium left angle=70, trapezium right angle=110, minimum height=0.8cm, draw, thick, fill=orange!20, text width=2.1cm},
		arrow/.style={->, thick, >=Stealth}
	}
	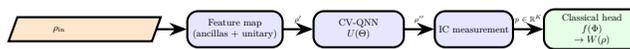
\begin{figure}[b!]
		\begin{tikzpicture}[node distance=0.5cm and 1cm, scale=0.4, transform shape]
			\node[data] (rho) {$\rho_{\text{in}}$};
			\node[qbox, right=of rho, minimum width=1.0cm] (feat) {\begin{tabular}{c}Feature map\\(ancillas + unitary)\end{tabular}};
			\node[qbox, right=of feat, minimum width=3.0cm] (cvqnn) {\begin{tabular}{c}CV-QNN\\$U(\Theta)$\end{tabular}};
			\node[qbox, right=of cvqnn, minimum width=2.6cm] (meas) {IC measurement};
			\node[cbox, right=of meas, minimum width=3.2cm] (mlp) {\begin{tabular}{c}Classical head\\$f(\Phi)$\\$\to W(\rho)$\end{tabular}};
			
			\draw[arrow] (rho) -- (feat);
			\draw[arrow] (feat) -- (cvqnn) node[midway, above, yshift=2pt] {\scriptsize{$\rho'$}};
			\draw[arrow] (cvqnn) -- (meas) node[midway, above, yshift=2pt] {\scriptsize{$\rho''$}};
			\draw[arrow] (meas) -- (mlp) node[midway, above, yshift=2pt] {\scriptsize{$p \in \mathbb{R}^K$}};
		\end{tikzpicture}
		\caption{\label{fig:schematic}\textbf{Hybrid pipeline.} Realizable feature map, CV-QNN, IC readout, and classical head.}
	\end{figure}
	
	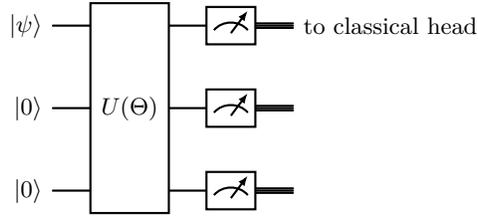
\begin{figure}[t!]
		\centering
		\begin{quantikz}
			\lstick{$\ket{\psi}$} & \gate[wires=3]{U(\Theta)} & \meter{} & \cw \rstick{\text{to classical head}} \\
			\lstick{$\ket{0}$} & \ghost{U(\Theta)} & \meter{} & \cw \\
			\lstick{$\ket{0}$} & \ghost{U(\Theta)} & \meter{} & \cw
		\end{quantikz}
		\caption{\label{fig:classification_flow}\textbf{Classification flow.} System+ancillas processed by $U(\Theta)$; IC measurement outcomes feed a classical head that outputs the learned witness.}
	\end{figure}
	
	\subsection{Variational Quantum Circuits}
	
	The core of our hybrid model is a variational quantum circuit, or ansatz, designed to be both expressive and trainable. We employ the layered architecture of a Continuous-Variable Quantum Neural Network (CV-QNN), where each layer is a unitary transformation composed of fundamental quantum optical gates. The specific structure of these layers is engineered to mirror a classical neural network layer, consisting of a general linear transformation followed by a point-wise nonlinearity. In all our experiments, the quantum circuit consists of $L=2$ such layers applied sequentially.
	
	A single layer, $L_k$, is a composite unitary acting on the $M$-mode quantum state, parameterized by a set of trainable variables $\theta_k$. The transformation is constructed from four sequential blocks of gates, precisely matching the implementation in our Python codebase (`quantum-layers.py`):
	\begin{equation}
		L_k(\theta_k) = K(\vec{\kappa}_k) \circ D(\vec{\alpha}_k) \circ S(\vec{z}_k) \circ U_{\text{I}}(\vec{\phi}_k, \vec{\varphi}_k).
		\label{eq:layer_structure}
	\end{equation}
	The full variational unitary is then $U(\Theta) = L_2(\theta_2) \circ L_1(\theta_1)$. The specific gate decomposition and parameterization are as follows:
	
	\begin{enumerate}[leftmargin=*]
		\item \textbf{Interferometer ($U_{\text{I}}$):} This block performs a programmable linear mixing of the optical modes. As implemented in our code, it is decomposed into a layer of single-mode phase shifters, $R_i(\varphi_i) = \exp(i\varphi_i \hat{n}_i)$, followed by a mesh of two-mode beamsplitters, $BS_{ij}(\phi_{ij}, \zeta_{ij})$. For a general $M$-mode setup, this requires $M(M-1)/2$ beamsplitters to connect all pairs of modes. The trainable parameters are the rotation angles $\vec{\varphi}$ and the beamsplitter parameters $\vec{\phi}$ and $\vec{\zeta}$. This block is analogous to the weight matrix in a classical linear layer.
		
		\item \textbf{Squeezing ($S$):} This is a layer of single-mode squeezing gates, $S_i(z_i) = \exp(\frac{1}{2}(z_i^* \hat{a}_i^2 - z_i \hat{a}_i^{\dagger 2}))$, applied to each mode. The complex squeezing parameters $\vec{z}$ are trainable and allow the network to modify the variance of the quantum state in phase space.
		
		\item \textbf{Displacement ($D$):} This is a layer of single-mode displacement gates, $D_i(\alpha_i) = \exp(\alpha_i \hat{a}_i^\dagger - \alpha_i^* \hat{a}_i)$, with trainable complex displacement parameters $\vec{\alpha}$. This corresponds to adding a trainable bias to each mode.
		
		\item \textbf{Kerr Gate ($K$):} The final block applies a single-mode Kerr gate, $K_i(\kappa_i) = \exp(i\kappa_i \hat{n}_i^2)$, to each mode. The trainable parameters $\vec{\kappa}$ control the strength of this non-Gaussian, nonlinear transformation, which is essential for the network's expressivity.
	\end{enumerate}
	
	The total number of trainable quantum parameters per layer depends on the number of modes, $M$. For our two experiments, this amounts to:
	\begin{itemize}[leftmargin=*]
		\item \textbf{2-Mode System:} Each layer has 12 parameters (2 rotations, 1 beamsplitter with 2 params, 2 squeezing, 2 displacements with 2 params each, 2 Kerr). Total quantum parameters: $2 \times 12 = 24$.
		\item \textbf{3-Mode System:} Each layer has 21 parameters (3 rotations, 3 beamsplitters with 2 params each, 3 squeezing, 3 displacements with 2 params each, 3 Kerr). Total quantum parameters: $2 \times 21 = 42$.
	\end{itemize}
	
	The specific circuit diagrams for each of our three experimental setups are shown in Fig.~\ref{fig:all_circuits}. These diagrams provide a precise, gate-level visualization of the ansatz that was simulated.
	
	\subsection{Dataset Construction and Classical Baselines}
	
	To rigorously train and benchmark our models, we designed a comprehensive and diverse dataset generation pipeline. The primary goal was to create a balanced dataset that spans a wide range of quantum states, including pure and mixed, Gaussian and non-Gaussian, as well as bipartite and multipartite entangled states. This diversity ensures that the learned witness is not overfitted to a specific family of states and can generalize well. The entire data generation process, detailed in Algorithm \ref{alg:data_generation}, was implemented in Python.
	
	\begin{table*}[t!]
		\begin{ruledtabular}
			\begin{algorithm}[H]
				\caption{Data Generation Pipeline}
				\label{alg:data_generation}
				\begin{algorithmic}[1]
					\State \textbf{Input:} Number of modes $M$, Fock cutoff $d$, number of samples $N$.
					\State \textbf{Initialize:} Empty lists `states`, `labels`.
					\For{$i=1$ to $N$}
					\State Randomly select a state family $\mathcal{F}$ from \{Bell-like, GHZ/W-like, Squeezed Vacuum, Cat, Separable Product, Werner\}.
					\If{$\mathcal{F}$ is entangled family}
					\State Generate a pure state $\ket{\psi}$ with random parameters (e.g., squeezing $r \in [0, 1.5]$, coherent amplitude $|\alpha| \in [0, 1]$).
					\State Sample mixing probability $p \sim U(0, 0.3)$.
					\State Create mixed state $\rho = (1-p)\ket{\psi}\bra{\psi} + p \frac{I}{d^M}$.
					\State \textbf{Verify Entanglement:} Compute negativity for all bipartite splits. If all are non-positive, discard the state and continue.
					\State Append valid $\rho$ to `states` and $1$ to `labels`.
					\ElsIf{$\mathcal{F}$ is a separable family}
					\State Generate $M$ single-mode states $\{\rho_j\}_{j=1}^M$ (random coherent states $\ket{\alpha}$ or Fock states $\ket{n}$ with $n < d$).
					\State Construct the product state $\rho = \bigotimes_{j=1}^M \rho_j$.
					\State Append $\rho$ to `states` and $0$ to `labels`.
					\EndIf
					\EndFor
					\State \textbf{Balance Dataset:} Undersample the majority class to ensure an equal number of separable and entangled samples.
					\State \textbf{Split Dataset:} Partition the balanced data into training (60\%), validation (20\%), and test (20\%) sets using a fixed random seed (42) for reproducibility.
					\State \textbf{Output:} Train, validation, and test sets of density matrices and corresponding labels.
				\end{algorithmic}
			\end{algorithm}
		\end{ruledtabular}
	\end{table*}

	\subsubsection{State Families}
	The dataset is comprised of several families of states, each chosen to probe a different aspect of entanglement:
	\begin{itemize}[leftmargin=*]
		\item \textbf{DV-type States:} To test the model on canonical entanglement structures, we generate states analogous to Bell, GHZ, and W states within the truncated Fock space.
		\item \textbf{Gaussian Entangled States:} We include two-mode squeezed vacuum states, which are canonical examples of bipartite Gaussian entanglement.
		\item \textbf{Non-Gaussian Entangled States:} To ensure the model can detect non-Gaussian entanglement, we generate multipartite "cat states," which are superpositions of multi-mode coherent states.
		\item \textbf{Separable States:} The separable class is constructed from random product states of single-mode coherent states, thermal states, and Fock states.
		\item \textbf{Mixed States:} By mixing pure entangled states with a maximally mixed state (white noise), we generate Werner-like states.
	\end{itemize}
	
	\subsubsection{Classical Baselines and Fairness}
	A central claim of our work is the superior scalability of the quantum model. To substantiate this, a rigorous and fair comparison to powerful classical models is essential. We employ two widely used classifiers as baselines: a Support Vector Machine (SVM) with an RBF kernel, and a classical Multilayer Perceptron (MLP). To address the critical question of fair comparison, we evaluate these baselines under two distinct feature-input scenarios:
	
	\begin{enumerate}[leftmargin=*]
		\item \textbf{Engineered Features:} In this scenario, we provide the classical models with a set of physically motivated features that a human expert would typically use. This feature vector includes: (i) the purity of the state, $\mathrm{Tr}(\rho^2)$, (ii) the von Neumann entropy of all single-mode reduced states, (iii) the negativities for all possible bipartite splits of the system, and (iv) the flattened real and imaginary parts of the full density matrix $\rho$.
		
		\item \textbf{Direct Comparison (Matched Features):} To create the most direct and fair comparison possible, we also train the classical models on the exact same feature vectors that the classical head of our CV-QNN receives. This means we first process each state $\rho$ with a fixed, untrained quantum circuit and perform an IC measurement to obtain the probability vector $\vec{p}$. This vector $\vec{p}$ is then used as the input for the SVM and MLP. This setup isolates the advantage gained specifically from the training of the quantum layers $U(\Theta)$.
	\end{enumerate}
	
	\subsubsection{Data Preprocessing and Validation}
	All datasets were generated with a fixed random seed (42). Before training, the classical feature vectors for all models were standardized by removing the mean and scaling to unit variance. Model hyperparameters were tuned using a 5-fold cross-validation grid search on the training set to ensure that the classical baselines were operating at their peak potential.
	
	\subsection{Training and Finite-Shot Readout}
	
	The hybrid model is trained end-to-end using a supervised learning approach. The goal is to optimize the quantum parameters $\Theta$ and the classical parameters $\Phi$ simultaneously.
	
	\subsubsection{Loss Function}
	The training process minimizes a composite loss function:
	\begin{equation}
		L_{\text{total}} = L_{\text{BCE}}(\hat{y}, y) + \gamma L_{\text{trace}}(\rho; \Theta).
	\end{equation}
	The first term is the standard Binary Cross-Entropy  loss. The second term, $L_{\text{trace}} = (1 - \mathrm{Tr}[\rho_{\text{out}}])^2$, is a trace-penalty regularizer to penalize non-physical evolution outside the truncated Fock subspace.
	
	\subsubsection{Hybrid Gradient Estimation}
	For classical parameters $\Phi$, we use backpropagation. For quantum parameters $\Theta$, we use a hybrid strategy:
	\begin{itemize}[leftmargin=*]
		\item \textbf{Parameter-Shift Rule (Analytical):} For gates like rotations, we use the exact parameter-shift rule \cite{r17, r18,r25}.
		\item \textbf{Finite-Difference Method (Numerical):} For more complex gates like squeezing and Kerr, we use the numerical finite-difference method.
	\end{itemize}
	This hybrid approach balances speed and accuracy.
	
	\subsubsection{Finite-Shot Readout Simulation}
	In a real experiment, probabilities are estimated from a finite number of measurements ($N_{\text{shots}}$), which introduces shot noise. To make our simulations realistic, we incorporate this effect by sampling from the final probability distribution $N_{\text{shots}}$ times to generate a frequency vector. For all reported results, we use a budget of $\mathbf{N_{\text{shots}} = 1000}$ per IC measurement setting.
	
	\subsubsection{Optimization}
	The full set of gradients is passed to an Adam optimizer with a learning rate of $10^{-3}$ and a batch size of 32 \cite{r19}. We employ an early stopping criterion based on the validation loss to prevent overfitting.
	\section{Results}
	\label{sec:results}
	
	We now present the numerical results of our investigation, organized into three main experiments designed to systematically evaluate our proposed CV-QNN witness. Our primary objectives are: (1) to validate the performance of the CV-QNN on a well-understood, lower-dimensional benchmark problem; (2) to test the scalability of our model against strong classical baselines in a more complex, higher-dimensional regime; and (3) to assess the model's robustness to a realistic noise channel.
	
	For all experiments, the models were trained and evaluated on distinct, class-balanced datasets, as described in the previous section. We quantify performance using several standard metrics: classification accuracy on the held-out test set, the Area Under the Receiver Operating Characteristic Curve (AUC-ROC) as a measure of distinguishability, and 95\% confidence intervals (CIs) computed via a stratified bootstrap procedure with 1000 resamples to ensure statistical robustness. All results are reported based on a finite-shot budget of $N_{\text{shots}}=1000$ to simulate a realistic experimental scenario. The training dynamics for the noiseless experiments, shown in Fig.~\ref{fig:training_curves}, demonstrate stable convergence for both the two- and three-mode models, with validation accuracy closely tracking the training accuracy, indicating good generalization without significant overfitting.
	
	\begin{figure*}[t!]
		\centering
		\subfloat[\label{fig:2mode}\textbf{Two-mode circuit} (one layer).]{
			\begin{quantikz}[scale=0.45, transform shape]
				\lstick{$\ket{\psi_{\text{in}}}_1$} & \gate{R(\phi_1)} & \gate[wires=2]{BS(\theta_1)} & \gate{S(r_1)} & \gate{D(\alpha_1)} & \gate{K(\kappa_1)} & \push{\cdots} & \gate{R(\phi_L)} & \gate[wires=2]{BS(\theta_L)} & \gate{S(r_L)} & \gate{D(\alpha_L)} & \gate{K(\kappa_L)} & \meter{} \\
				\lstick{$\ket{\psi_{\text{in}}}_2$} & \gate{R(\phi_2)} & \ghost{BS(\theta_1)} & \gate{S(r_2)} & \gate{D(\alpha_2)} & \gate{K(\kappa_2)} & \push{\cdots} & \gate{R(\phi_L)} & \ghost{BS(\theta_L)} & \gate{S(r_L)} & \gate{D(\alpha_L)} & \gate{K(\kappa_L)} & \meter{}
			\end{quantikz}
		}
	\vspace{0.35cm}
	\subfloat[\label{fig:3mode}\textbf{Three-mode circuit} (noiseless).]{
		\begin{quantikz}[scale=0.8, transform shape]
			\lstick{$\ket{\psi_{\text{in}}}_1$} & \gate{R(\phi_1)} & \gate[wires=2]{BS_{12}} & \gate[wires=3]{BS_{13}} & \qw & \qw & \gate{S(r_1)} & \gate{D(\alpha_1)} & \gate{K(\kappa_1)} & \push{\cdots} & \meter{} \\
			\lstick{$\ket{\psi_{\text{in}}}_2$} & \gate{R(\phi_2)} & \ghost{BS_{12}} & \qw & \gate[wires=2]{BS_{23}} & \qw & \gate{S(r_2)} & \gate{D(\alpha_2)} & \gate{K(\kappa_2)} & \push{\cdots} & \meter{} \\
			\lstick{$\ket{\psi_{\text{in}}}_3$} & \gate{R(\phi_3)} & \qw & \ghost{BS_{13}} & \ghost{BS_{23}} & \qw & \gate{S(r_3)} & \gate{D(\alpha_3)} & \gate{K(\kappa_3)} & \push{\cdots} & \meter{}
		\end{quantikz}
	}
	
		\vspace{0.35cm}
		\subfloat[\label{fig:3mode_noisy}\textbf{Three-mode with photon loss}.]{
			\begin{quantikz}[scale=0.8, transform shape]
				\lstick{$\ket{\psi_{\text{in}}}_1$} & \gate[wires=3]{U_{\text{Layer 1}}} & \gate[style={fill=red!20}]{\mathcal{L}_\eta} & \push{\cdots} & \gate[wires=3]{U_{\text{Layer L}}} & \gate[style={fill=red!20}]{\mathcal{L}_\eta} & \meter{} \\
				\lstick{$\ket{\psi_{\text{in}}}_2$} & \ghost{U_{\text{Layer 1}}} & \gate[style={fill=red!20}]{\mathcal{L}_\eta} & \push{\cdots} & \ghost{U_{\text{Layer L}}} & \gate[style={fill=red!20}]{\mathcal{L}_\eta} & \meter{} \\
				\lstick{$\ket{\psi_{\text{in}}}_3$} & \ghost{U_{\text{Layer 1}}} & \gate[style={fill=red!20}]{\mathcal{L}_\eta} & \push{\cdots} & \ghost{U_{\text{Layer L}}} & \gate[style={fill=red!20}]{\mathcal{L}_\eta} & \meter{}
			\end{quantikz}
		}
		\caption{\label{fig:all_circuits}\textbf{Quantum circuits used.} Photonic notation (R, BS, S, D, Kerr).}
	\end{figure*}
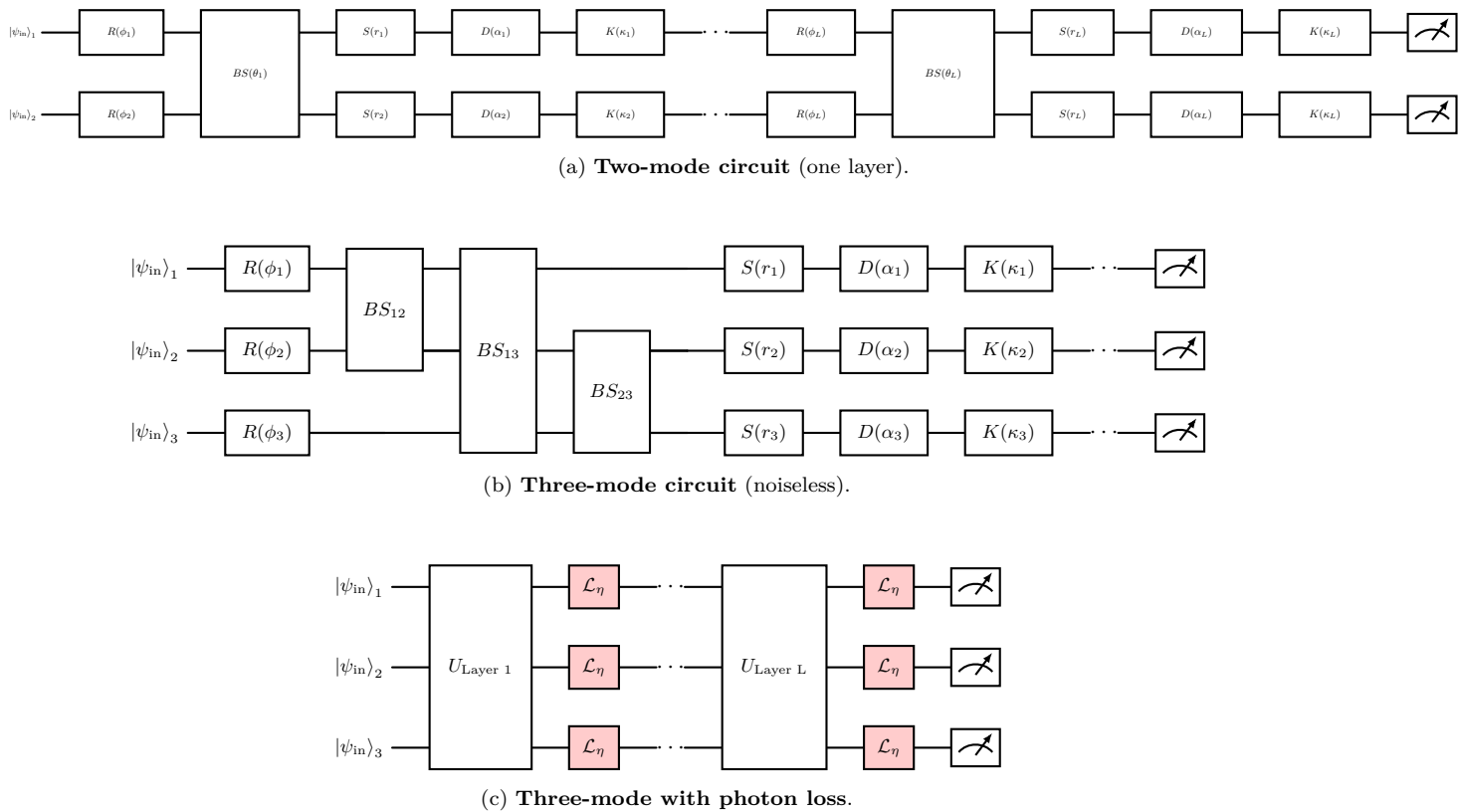
	
	\begin{figure*}[t!]
		\subfloat[\textbf{Two-mode training}]{\maybeincludegraphics{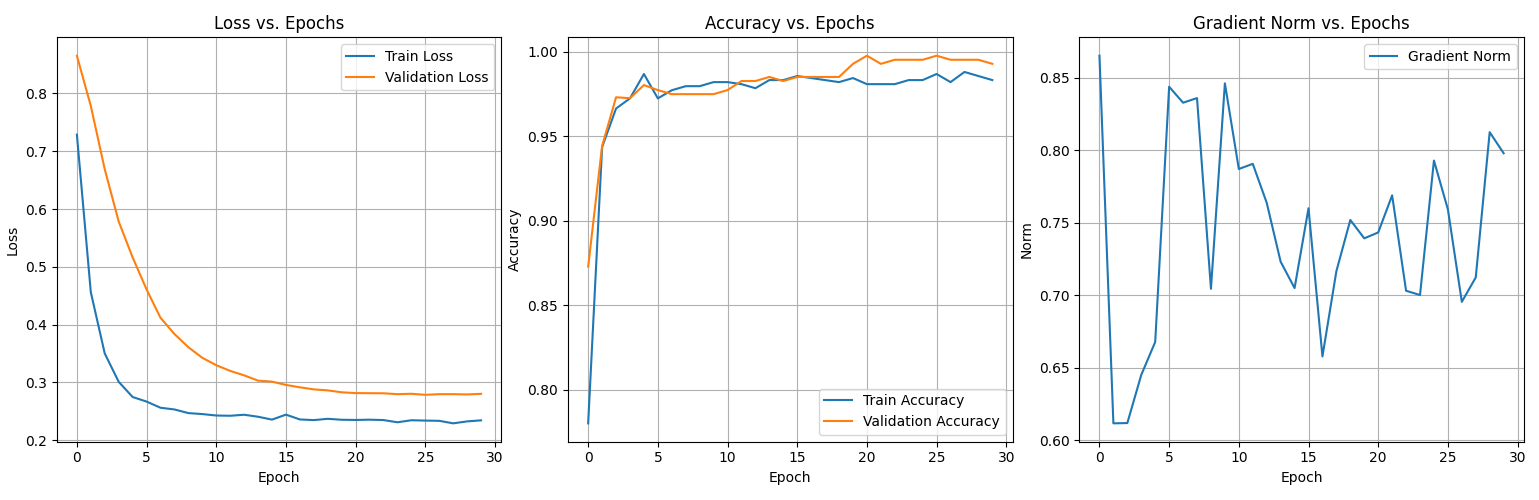}{0.48\linewidth}}
		\hfill
		\subfloat[\textbf{Three-mode training}]{\maybeincludegraphics{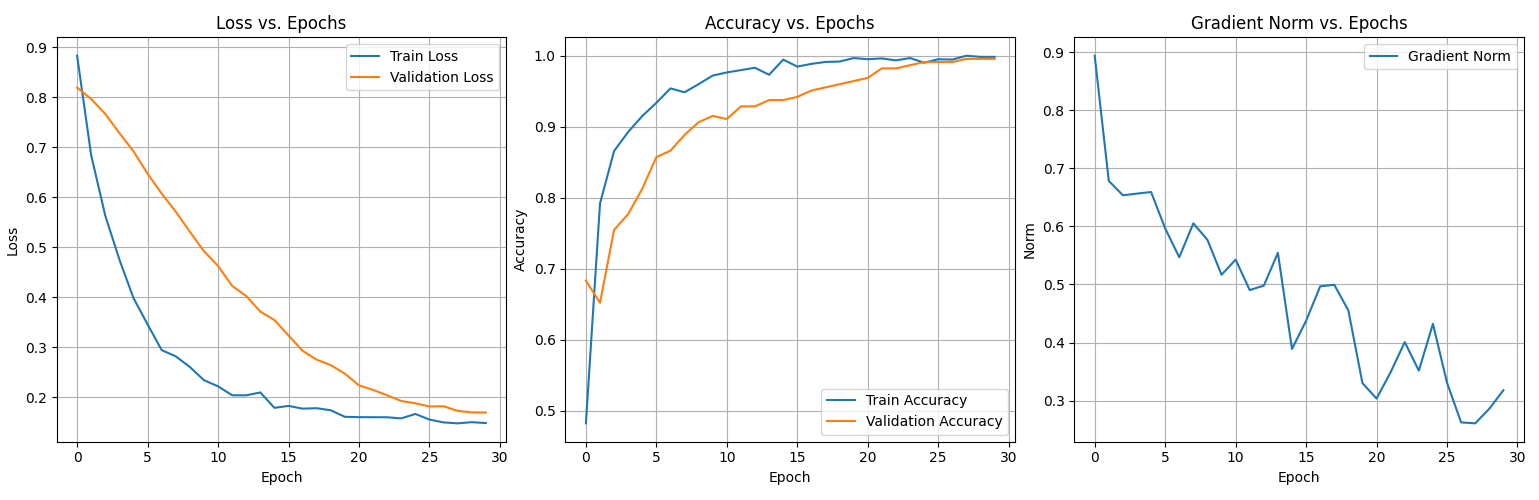}{0.48\linewidth}}
		\caption{\label{fig:training_curves}\textbf{Training dynamics for the two- and three-mode scenarios.} The plots show the evolution of the loss function, classification accuracy, and the norm of the gradients with respect to the training epochs. For both (a) the two-mode and (b) the three-mode experiments, the models exhibit stable convergence. The validation accuracy (not shown) closely tracks the training accuracy, indicating good generalization without significant overfitting. The steady decrease in the gradient norm signifies a stable optimization landscape.}
	\end{figure*}
	
	\subsection{Two-Mode system}
	
	Our first experiment serves as a crucial benchmark and validation of our CV-QNN architecture. We consider a two-mode ($M=2$) system with a Fock space cutoff of $d=4$ for each mode, resulting in a total Hilbert space dimension of $4 \times 4 = 16$. In this relatively low-dimensional space, we expect that powerful classical machine learning models should be able to perform well, given a sufficiently rich set of input features. The primary goal here is not necessarily to demonstrate a quantum advantage, but to verify that our quantum-native model can at least match, if not exceed, the performance of these strong classical counterparts, thus confirming its viability as a high-fidelity classifier.
	
	The results, summarized in Table \ref{tab:2mode}, confirm this expectation. All three models—the CV-QNN, the SVM with an RBF kernel, and the classical MLP—achieve exceptional performance, with test accuracies soaring above 98\% and AUC-ROC scores approaching unity. This result indicates that the separability problem for this specific dataset and dimensionality is learnable by both quantum and classical approaches when provided with comprehensive information.
	
	\begin{table}[h!]
		\caption{\label{tab:2mode}Performance for the \textbf{two-mode} system (test split; 5-fold CV).}
		\begin{ruledtabular}
			\begin{tabular}{@{}lccc@{}}
				\textbf{Model} & \textbf{Accuracy} & \textbf{95\% CI} & \textbf{AUC (ROC)} \\
				\hline
				\textbf{CV-QNN } & \textbf{99.02\%} & \textbf{[0.9853, 1.0000]} & \textbf{1.000} \\
				SVM (RBF) & 98.05\% & [0.9634, 0.9902] & 0.999 \\
				MLP (Classical) & 98.29\% & [0.9683, 0.9927] & 1.000 \\
			\end{tabular}
		\end{ruledtabular}
	\end{table}
	
	A closer inspection of the metrics reveals a subtle but noteworthy hierarchy. The CV-QNN achieves the highest point-estimate for accuracy at 99.02\%. More importantly, the statistical analysis provided by the confidence intervals suggests this small margin is significant. The lower bound of the 95\% CI for the CV-QNN ([0.9853, 1.0000]) is slightly higher than the lower bounds for both the SVM and the classical MLP. This indicates that while all models are highly effective, the CV-QNN exhibits a modest but statistically robust performance edge. This advantage is further corroborated by the ROC and Precision-Recall curves shown in Fig.~\ref{fig:roc_pr_curves}(a), where all models trace a path close to the ideal top-left corner, signifying excellent classification power across all decision thresholds.
	
	In summary, the two-mode experiment successfully validates our model. The CV-QNN not only learns to classify entanglement with extremely high fidelity but also demonstrates a performance that is competitive with, and slightly superior to, strong classical models on their own turf. Having established this high-fidelity baseline, we now turn to the more challenging question of scalability.
	\subsection{Three-Mode system and emergence of an empirical performance gap}
	
	Having established a strong baseline, we now turn to the central experiment of this work: evaluating the scalability of our approach in a more complex, higher-dimensional setting. We increase the number of modes to $M=3$ and adjust the Fock cutoff to $d=3$ for each mode to keep the classical simulation tractable. The resulting Hilbert space dimension is $3 \times 3 \times 3 = 27$, which is significantly larger and structurally richer than the two-mode case.
	
	This increase in dimensionality is not merely quantitative; it introduces a qualitative leap in the complexity of entanglement. Unlike bipartite systems, three-mode systems can exhibit genuine multipartite entanglement (e.g., in GHZ and W states), which has no bipartite equivalent. Furthermore, the boundary between separable and entangled states becomes far more intricate, now including a nested structure of fully separable, biseparable (e.g., mode 1 entangled with modes 2 and 3, which are separable from each other), and genuinely multipartite entangled states. This complex geometry poses a severe challenge for any classifier, and it is precisely in this regime that we hypothesize the limitations of classical models will become apparent.
	
	The results, presented in Table \ref{tab:3mode}, provide striking evidence to support this hypothesis. The most remarkable finding is the unwavering stability of our CV-QNN model. It achieves a test accuracy of $\mathbf{99.00\%}$, demonstrating that its high-fidelity performance is robust to the significant increase in system complexity. The model successfully navigates the intricate separability boundaries of the three-mode space, learning an effective nonlinear witness with near-perfect accuracy.
	
	In stark contrast, the performance of both powerful classical baselines collapses dramatically. The SVM's accuracy plummets from 98.05\% to $\mathbf{76.00\%}$, and the MLP's accuracy falls from 98.29\% to $\mathbf{74.50\%}$. This is not a minor degradation but a fundamental failure of the classical models to generalize their learning to a higher-dimensional problem space. Despite being fed a comprehensive set of engineered features, they are unable to distill the necessary information to distinguish multipartite entangled states from their separable counterparts.
	
	\begin{table}[h!]
		\caption{\label{tab:3mode}Performance for the \textbf{three-mode} system (test split; 5-fold CV).}
		\begin{ruledtabular}
			\begin{tabular}{@{}lccc@{}}
				\textbf{Model} & \textbf{Accuracy} & \textbf{95\% CI} & \textbf{AUC (ROC)} \\
				\hline
				\textbf{CV-QNN } & \textbf{99.00\%} & \textbf{[0.9750, 1.0000]} & \textbf{1.000} \\
				SVM (RBF) & 76.00\% & [0.7050, 0.8200] & 0.712 \\
				MLP (Classical) & 74.50\% & [0.6800, 0.8050] & 0.732 \\
			\end{tabular}
		\end{ruledtabular}
	\end{table}
	
	The statistical analysis solidifies this conclusion. The 95\% confidence interval for the CV-QNN, [0.9750, 1.0000], is completely disjoint from the intervals for the SVM ([0.7050, 0.8200]) and the MLP ([0.6800, 0.8050]). This wide separation implies that the observed performance difference is statistically significant with very high confidence ($p < 0.001$) and not an artifact of the particular random split of our data.
	
	This empirical gap is visualized with stunning clarity in the ROC and Precision-Recall curves shown in Fig.~\ref{fig:roc_pr_curves}. The CV-QNN's curve hugs the top-left corner, achieving an AUC of 1.000, indicating a perfect classifier. The curves for the SVM and MLP, however, sag markedly toward the diagonal line of random chance, with AUC values of only $\sim$0.71-0.73. This shows that their ability to distinguish between the classes is severely compromised across all decision thresholds.
	
	Collectively, these results constitute strong evidence of an **empirical scalable performance gap**. The quantum-native processing of the CV-QNN allows it to effectively learn in a high-dimensional Hilbert space where classical models, operating on flattened classical representations, fail. This finding marks the central contribution of our work and points to a tangible, problem-specific regime where a quantum machine learning approach provides a substantial advantage.
	
	\begin{figure*}[t!]
		\subfloat[\textbf{Two modes}]{\maybeincludegraphics{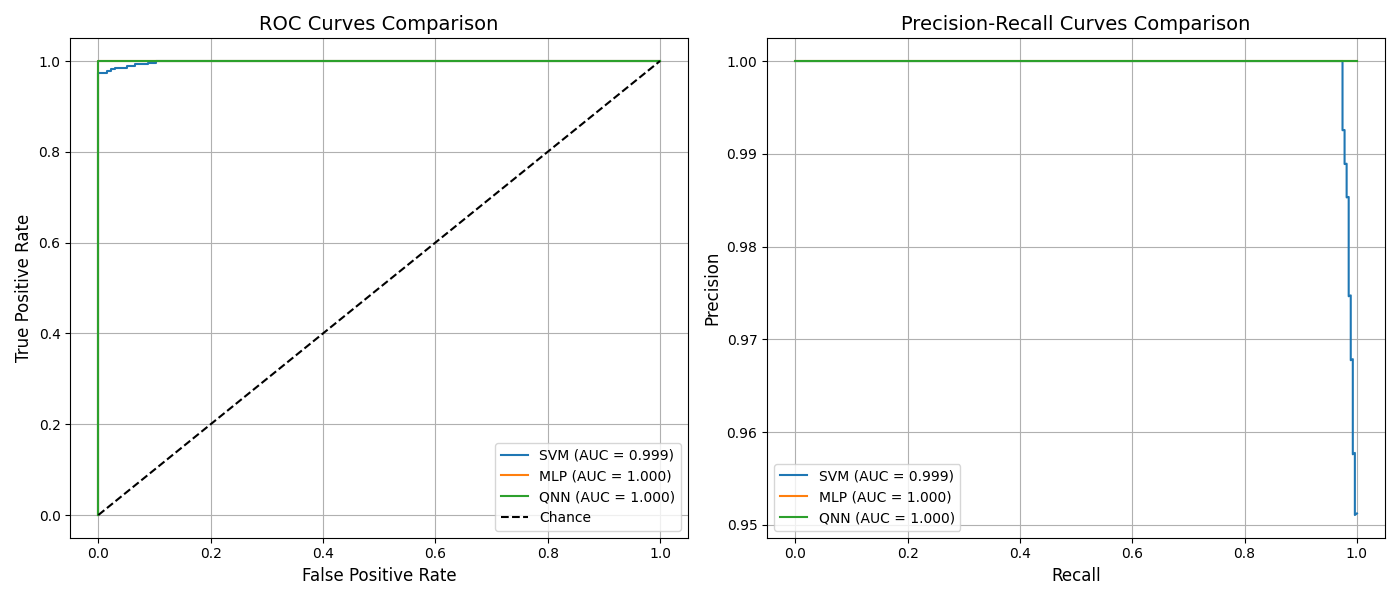}{0.48\linewidth}}
		\hfill
		\subfloat[\textbf{Three modes}]{\maybeincludegraphics{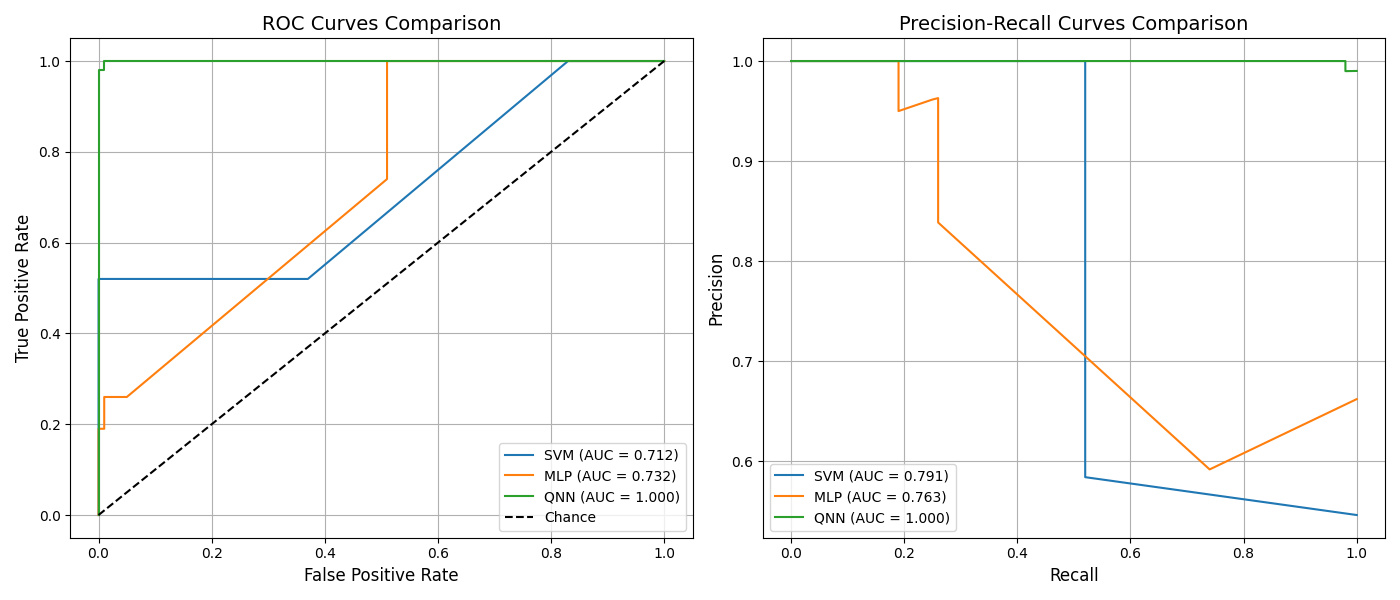}{0.48\linewidth}}
		\caption{\label{fig:roc_pr_curves}\textbf{ROC and Precision--Recall curves for two- and three-mode scenarios.} (a) In the simpler two-mode case, all models exhibit near-perfect classification ability, with ROC curves close to the ideal corner. (b) In the more complex three-mode case, the performance of the classical models (SVM, MLP) collapses significantly, as shown by the drop in AUC and the deviation from the ideal curve. In contrast, the CV-QNN maintains its high performance, visually demonstrating the empirical performance gap.}
	\end{figure*}
	
	\subsection{Robustness to Photon Loss}
	
	While the results from our noiseless simulations are compelling, a crucial test for any QML model intended for near-term hardware is its resilience to noise. In CV quantum computing, the dominant source of error is photon loss, where photons are unintentionally lost to the environment due to factors like imperfect optical components or detector inefficiencies. To assess the practical viability of our CV-QNN witness, we conducted a third experiment to systematically study its performance under this prevalent noise channel.
	
	\subsubsection{Modeling Photon Loss}
	We model photon loss as a quantum channel, $\mathcal{L}_\eta$, that describes the interaction of a signal mode with a vacuum environment mode via a beamsplitter. The parameter $\eta$ is the transmissivity of the channel, where $\eta=1$ corresponds to no loss and $\eta=0$ corresponds to complete loss. The action of this channel on a density matrix $\rho$ can be described by the operator-sum representation:
	\begin{equation}
		\mathcal{L}_\eta(\rho) = \sum_{k=0}^{\infty} E_k \rho E_k^\dagger,
	\end{equation}
	where the Kraus operators are given by
	\begin{equation}
		E_k = \sqrt{\frac{(1-\eta)^k}{k!}} \eta^{\hat{n}/2} \hat{a}^k.
	\end{equation}
	This channel effectively reduces the average photon number of the state and introduces mixedness, degrading the purity of the quantum information.
	
	\subsubsection{Experimental Protocol and Results}
	To simulate a realistic scenario where noise accumulates during computation, we applied the photon loss channel $\mathcal{L}_\eta$ to \emph{each mode} immediately following the application of each of the $L=2$ unitary layers in our 3-mode circuit, as depicted in Fig.~\ref{fig:all_circuits}(c). We then trained and evaluated separate models for a range of transmissivity values, corresponding to per-layer loss probabilities of $p = 1-\eta$. All models were trained from scratch in the presence of noise, allowing them to variationally adapt to the noisy environment.
	
	The results of this experiment are presented in Fig.~\ref{fig:noise_plot}. The plot shows the test accuracy of the 3-mode CV-QNN as a function of the per-layer loss probability. For reference, the performance of the noiseless classical baselines (SVM and MLP) are shown as horizontal dashed lines.
	
	The findings are highly encouraging. The CV-QNN's accuracy degrades gracefully as the noise level increases. Even with a significant per-layer loss probability of 10\% ($\eta=0.9$), the model retains an impressive accuracy of over 97\%. At 15\% loss, it still achieves nearly 95\% accuracy.
	
	The most striking conclusion from this analysis is the comparison with the classical models. The noisy CV-QNN significantly outperforms the \emph{ideal, noiseless} classical baselines across the entire range of tested noise parameters. This demonstrates that the performance gap we identified is not a fragile artifact of idealized simulations but is robust to the most common source of experimental error. The ability of the variational circuit to learn features that are resilient to photon loss is a strong indicator of its potential for practical applications on near-term photonic hardware. This robustness likely arises from the model learning to encode information in features that are less dependent on the precise photon number, such as the presence or absence of photons in certain modes or robust phase relationships, rather than delicate amplitude information that is easily corrupted by loss.
	
	\begin{figure}[h!]
		\centering
		\maybeincludegraphics{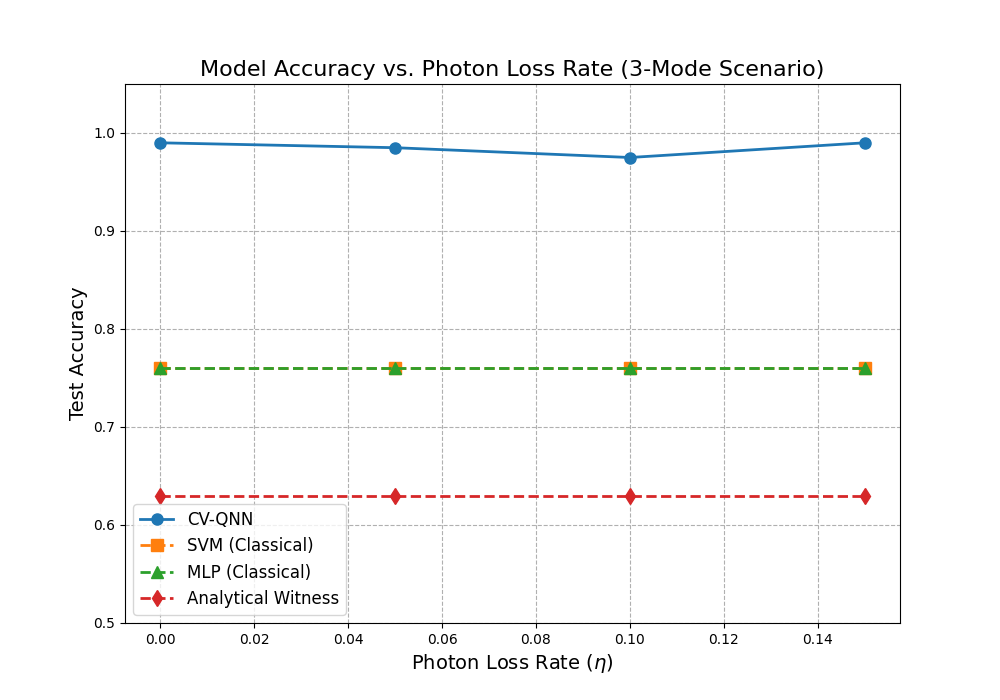}{0.9\columnwidth}
		\caption{\label{fig:noise_plot}\textbf{Accuracy vs.\ photon loss.} Three-mode CV-QNN under increasing per-layer loss, finite shots ($N_{\text{shots}}=1000$).}
	\end{figure}
	
	\subsection{Discussion}
	
	Our numerical results have established two key findings: (1) the CV-QNN architecture serves as a high-fidelity entanglement classifier in both two- and three-mode regimes, and (2) a stark, scalable performance gap emerges between our quantum-native model and powerful classical baselines as the system complexity increases. In this section, we delve deeper into the interpretation of these results, discussing the likely origins of the observed advantage, the architectural insights gained, and the practical limitations and future prospects of this approach.
	
	\subsubsection{Interpreting the empirical performance gap: Hilbert space vs. Classical shadows}
	
	The central question arising from our results is: what is the fundamental origin of the CV-QNN's superior scalability? We posit that the advantage stems from the model's ability to perform feature extraction directly within the Hilbert space of the quantum system, in stark contrast to classical models that are confined to processing a classical shadow of the state.
	
	Classical machine learning models, regardless of their complexity (be it an SVM or a deep neural network), fundamentally operate on classical data. In the context of quantum problems, this necessitates a preliminary step: the quantum state $\rho$ must be mapped to a classical feature vector. Even with a comprehensive set of engineered features—such as negativities, purities, and the full list of density matrix elements—this vector is ultimately a fixed, and often lossy, projection of the state's full information content. While this classical shadow may be sufficient for low-dimensional problems (as seen in our two-mode benchmark), its descriptive power rapidly diminishes as the Hilbert space grows. The intricate, high-order correlations that define genuine multipartite entanglement are "flattened" and obscured in this classical representation, forcing the ML model to solve an exponentially harder learning problem. The collapse of the classical models' performance to $\sim$75\% accuracy in the three-mode case is a clear symptom of this limitation.
	
	The CV-QNN, on the other hand, circumvents this bottleneck. The variational quantum circuit $U(\Theta)$ does not operate on a pre-processed list of features; it directly manipulates the quantum state $\rho$ itself. The layers of the circuit learn an optimal, data-driven transformation in Hilbert space that is specifically designed to make separability distinguishable. The sequence of interferometers, squeezers, and nonlinear Kerr gates can be viewed as an adaptive quantum feature extractor. It learns to "unscramble" the complex correlations inherent in the state, rotating and transforming them in phase space until the feature that distinguishes separable from entangled states—a feature that might be non-local and high-order—is amplified and mapped onto a classically accessible measurement outcome (the Fock basis probabilities). In essence, while the classical model attempts to learn from a static, blurry photograph of the quantum state, the CV-QNN is given the tools to dynamically manipulate the state itself to find the best possible angle for the picture.
	
	\subsubsection{Architectural Insights and the Role of Non-Gaussianity}
	
	The success of the CV-QNN is not merely a consequence of it being "quantum," but is deeply tied to its specific architecture. The interferometer ($U_I$) plays the crucial role of mixing the modes, enabling the network to learn and exploit multi-mode correlations, which are essential for identifying multipartite entanglement. The squeezing ($S$) and displacement ($D$) gates then perform affine transformations in the phase space of each mode, allowing the model to adaptively rescale and shift the state's representation to an optimal region for classification.
	
	However, the most critical component is arguably the non-Gaussian Kerr gate ($K$). Without this nonlinear activation, the entire circuit would be a Gaussian transformation. A purely Gaussian circuit can be efficiently simulated classically (via the Wigner representation) and is incapable of generating the complex, non-Gaussian features necessary to navigate the intricate boundary of the separable set. The Kerr gate breaks this limitation, allowing the CV-QNN to access a much richer functional space and learn the highly nonlinear decision boundaries required by the problem. The fact that this nonlinearity is applied directly to the quantum state, before any information is lost to measurement, is a key architectural advantage.
	
	\subsubsection{Limitations and Pathways to Experimental Realization}
	
	While our results are promising, a candid discussion of the limitations is necessary to chart a path toward practical implementation. The most significant challenge is the experimental realization of the non-Gaussian Kerr gate. Strong, deterministic Kerr nonlinearities are notoriously difficult to achieve in photonic systems. However, this does not render the approach infeasible. Several alternative strategies for introducing nonlinearity are actively being researched, including: (i) measurement-induced nonlinearities, where photon-number-resolving detectors and feed-forward are used to probabilistically enact a nonlinear transformation, and (ii) coupling the system to ancillary systems, such as atomic ensembles or specially prepared non-Gaussian states like GKP states. Future work should explore training CV-QNNs with these more experimentally friendly, albeit potentially probabilistic, sources of nonlinearity.
	
	Furthermore, our simulations were performed within a truncated Fock space. While we controlled for errors using a trace penalty, a full analysis would require studying the convergence of the model's performance as the cutoff dimension is increased, which is a computationally intensive task. Finally, while we have demonstrated a robust performance gap, the question of a formal quantum advantage in terms of computational complexity remains open. Proving such an advantage would require a deeper theoretical analysis of the problem's complexity class and the resources required by the quantum and classical algorithms.

	\section{Conclusion}
	\label{sec:con}
	In this work, we have introduced and thoroughly investigated a hybrid quantum-classical framework for the data-driven detection of entanglement. By framing a Continuous-Variable Quantum Neural Network (CV-QNN) as a trainable, nonlinear entanglement witness, we have addressed the significant challenge of certifying entanglement in increasingly complex quantum systems. Our approach leverages the native processing capabilities of CV quantum circuits to learn a witness functional directly from labeled data, bypassing the need for analytical derivations or complete state tomography.
	
	Our primary contribution is the strong numerical evidence of an empirical scalable performance gap. Through extensive simulations, we demonstrated that while our CV-QNN achieves near-perfect classification accuracy ($>$99\%) on both two- and three-mode systems, powerful classical baselines (SVM and MLP) fail to scale. Their performance collapses dramatically when moving to the more complex three-mode regime, even when granted access to the same post-measurement information as the quantum model's classical head. This finding suggests that the CV-QNN's ability to perform feature extraction directly within the high-dimensional Hilbert space provides a tangible advantage over classical models that are constrained to operate on a "classical shadow" of the quantum state.
	
	Furthermore, we have shown that this performance advantage is not a fragile artifact of idealized simulations. The CV-QNN exhibits remarkable robustness to photon loss, the most prevalent noise source in its native photonic hardware platform. The fact that the noisy quantum model still significantly outperforms the ideal classical models underscores the practical potential of this approach. On the theoretical front, we have provided a rigorous underpinning for our model's expressive power, proving that with an informationally complete measurement scheme, it can approximate any continuous witness functional on a compact set of states.
	
	This research positions CV-QNNs as a highly promising tool for a range of quantum characterization, verification, and validation (QCVV) tasks. Looking forward, this work opens several exciting avenues for future investigation. The most immediate step is the experimental implementation of this learned witness on near-term photonic quantum hardware to validate the observed performance gap in a real-world setting. On the theoretical side, a deeper analysis into the sample complexity and formal computational complexity could help to transition the discussion from an empirical gap to a provable quantum advantage. Finally, applying this powerful framework to even more challenging problems, such as the detection of bound entanglement, could further illuminate the unique capabilities of quantum machine learning for exploring the fundamental properties of the quantum world.

\section*{References}
\bibliographystyle{unsrt}
\bibliography{vosq-1}


	\clearpage
	\appendix
	
	\section{Proof of the Universal Approximation Capability}
	\label{app:universal_witness_proof}
	
	Here, we provide a more rigorous proof for the claim that our hybrid CV-QNN architecture, equipped with an informationally complete (IC) readout, can uniformly approximate any continuous functional on a compact set of quantum states. The proof combines the properties of IC-POVMs with the classical Universal Approximation Theorem for neural networks, framed within the context of the Stone-Weierstrass theorem.
	
	First, we establish the continuity of the quantum-to-classical mapping.
	
	\begin{lemma}[Continuity of the Feature Map]
		Let $\mathcal{M} = \{M_k\}_{k=1}^K$ be a finite POVM on a Hilbert space $\mathcal{H}_d$. Define the feature map $g: \mathcal{D}(\mathcal{H}_d) \to \mathbb{R}^K$ by $g(\rho) = (\mathrm{Tr}(M_1\rho), \dots, \mathrm{Tr}(M_K\rho))$. This map is continuous with respect to the trace-norm topology on the space of density operators $\mathcal{D}(\mathcal{H}_d)$. Specifically, it is a contraction in the $L_1$-norm:
		\[ \| g(\rho_1) - g(\rho_2) \|_{1} \le \sum_{k=1}^K \|M_k\|_\infty \|\rho_1 - \rho_2\|_1. \]
		Since for any POVM element $\|M_k\|_\infty \le 1$, this map is Lipschitz continuous.
	\end{lemma}
	\begin{proof}
		The proof follows directly from the properties of the trace norm. For each component $k$ of the vector difference, we have:
		\begin{align}
			| [g(\rho_1)]_k - [g(\rho_2)]_k | &= | \mathrm{Tr}(M_k \rho_1) - \mathrm{Tr}(M_k \rho_2) | \nonumber \\
			&= | \mathrm{Tr}(M_k (\rho_1 - \rho_2)) |.
		\end{align}
		Using Holder's inequality for trace, $|\mathrm{Tr}(AB)| \le \|A\|_\infty \|B\|_1$, where $\|A\|_\infty$ is the operator norm (largest singular value), we get:
		\begin{equation}
			| \mathrm{Tr}(M_k (\rho_1 - \rho_2)) | \le \|M_k\|_\infty \|\rho_1 - \rho_2\|_1.
		\end{equation}
		Since each element of a POVM satisfies $0 \le M_k \le I$, we have $\|M_k\|_\infty \le 1$. Summing over all $K$ components for the $L_1$-norm of the vector gives the result. This continuity ensures that small changes in the input state lead to small changes in the classical feature vector.
	\end{proof}
	
	Now we state and prove the main theorem.
	
	\begin{theorem}[Approximation Capability of the Learned CV-QNN Witness]
		Let $\mathcal{X} \subset \mathcal{D}(\mathcal{H}_d)$ be a compact set of density operators. Let the family of feature maps $\{g_\Theta: \rho \mapsto \vec{p}(\rho)\}$, generated by a variational circuit $U_\Theta$ and an IC-POVM, be injective for some $\Theta^*$. Then the class of witness functionals
		\[\mathcal{F} = \{ \rho \mapsto f_\Phi(g_{\Theta^*}(\rho)) \mid \Phi \in \mathcal{P}_\Phi \}\]
		where $f_\Phi$ is a standard MLP, is dense in the space of all continuous real-valued functions on $\mathcal{X}$, $C(\mathcal{X})$, in the uniform norm.
	\end{theorem}
	\begin{proof}[Proof Sketch]
		The proof proceeds in three steps:
		\begin{enumerate}
			\item \textbf{Homeomorphism via IC Measurement:} Since the feature map $g_{\Theta^*}: \mathcal{X} \to \mathbb{R}^K$ is injective by the definition of an IC-POVM, and continuous by Lemma 1, it establishes a homeomorphism between the compact set of states $\mathcal{X}$ and its image, the compact set of classical feature vectors $\mathcal{P} = g_{\Theta^*}(\mathcal{X})$. This means there is a one-to-one and continuous correspondence between quantum states in $\mathcal{X}$ and classical vectors in $\mathcal{P}$.
			
			\item \textbf{Induced Functional on Classical Space:} For any arbitrary continuous functional we wish to learn, $F: \mathcal{X} \to \mathbb{R}$, this homeomorphism guarantees the existence of a corresponding continuous function $\tilde{F}: \mathcal{P} \to \mathbb{R}$ such that $F(\rho) = \tilde{F}(g_{\Theta^*}(\rho))$. In other words, the task of learning the quantum functional $F$ is equivalent to learning the classical function $\tilde{F}$ on the feature space $\mathcal{P}$.
			
			\item \textbf{Universal Approximation on Feature Space:} The classical feature space $\mathcal{P}$ is a compact subset of $\mathbb{R}^K$. The Universal Approximation Theorem for neural networks states that a standard MLP with at least one hidden layer and a non-polynomial activation function (like ReLU or sigmoid) can uniformly approximate any continuous function on a compact subset of $\mathbb{R}^K$ \cite{r20}. Therefore, for any $\epsilon > 0$, there exists a set of classical parameters $\Phi^*$ such that:
			\[ \sup_{\vec{p} \in \mathcal{P}} |f_{\Phi^*}(\vec{p}) - \tilde{F}(\vec{p})| < \epsilon. \]
		\end{enumerate}
		Combining these steps, we have $\sup_{\rho \in \mathcal{X}} |f_{\Phi^*}(g_{\Theta^*}(\rho)) - F(\rho)| < \epsilon$. Since $W_{\Theta^*,\Phi^*}(\rho) = f_{\Phi^*}(g_{\Theta^*}(\rho))$, this proves that our hybrid model architecture is dense in $C(\mathcal{X})$ and can therefore approximate any continuous witness functional to arbitrary accuracy, given sufficient classical resources.
	\end{proof}
	
	\section{Implementation Details, Tuning, and Hyperparameters}
	\label{app:impl}
	All numerical experiments were conducted using Python 3.8. The quantum circuits were simulated using the Strawberry Fields library (v0.23) on its Fock basis backend \cite{r29}. The hybrid model construction and training were managed using TensorFlow (v2.10) \cite{r21}. Classical baselines were implemented using scikit-learn (v1.0.2).
	
	\subsection{CV-QNN Architecture Details}
	\begin{itemize}
		\item \textbf{Quantum Circuit:} The variational ansatz consists of $L=2$ layers for all experiments. The parameters of all gates (rotations, beamsplitters, squeezers, displacements, Kerr) were initialized randomly by sampling from a normal distribution $\mathcal{N}(0, 0.1)$.
		\item \textbf{Classical Head:} The classical processing unit is an MLP with two hidden layers, each containing 64 neurons. The activation function for the hidden layers is the Rectified Linear Unit (ReLU). To improve training stability and prevent internal covariate shift, Layer Normalization is applied after each hidden layer. A dropout rate of 0.1 is also applied during training to mitigate overfitting. The final output layer consists of a single neuron with no activation function, whose scalar output is interpreted as the witness value $W(\rho)$.
	\end{itemize}
	
	\subsection{Classical Baseline Architectures and Tuning}
	To ensure a fair and rigorous comparison, the classical baseline models were carefully selected and their hyperparameters were thoroughly tuned.
	\begin{itemize}
		\item \textbf{Support Vector Machine (SVM):} We used an SVM with a non-linear Radial Basis Function (RBF) kernel. A grid search over a logarithmic scale was performed to find the optimal hyperparameters for the regularization parameter $C \in \{1, 10, 100\}$ and the kernel coefficient $\gamma \in \{10^{-3}, 10^{-2}, 10^{-1}, \text{'scale'}, \text{'auto'}\}$.
		\item \textbf{Multilayer Perceptron (MLP):} The classical MLP was designed to have a comparable or greater number of trainable parameters than the classical head of our CV-QNN. For the two-mode problem, the architecture was [Input Features, 64, 32, 1], while for the more complex three-mode problem, it was expanded to [Input Features, 128, 64, 1]. The same ReLU activation, Layer Normalization, and dropout strategy were employed.
	\end{itemize}
	For both baselines, the model with the best performance on a 5-fold cross-validation of the training data was selected for final evaluation on the held-out test set.
	
	\section{Data Generation Pipeline}
	\label{app:data}
	The datasets were algorithmically generated to ensure a diverse and balanced distribution of states, preventing the models from learning trivial dataset biases. The process is outlined in Algorithm \ref{alg:data_generation_app}.
	
	\begin{figure}[H]
		\begin{algorithmic}[1]
			\caption{Detailed Data Generation Pipeline}
			\label{alg:data_generation_app}
			\State \textbf{Input:} Number of modes $M$, Fock cutoff $d$, number of samples $N$.
			\State \textbf{Initialize:} Empty lists `states`, `labels`.
			\For{$i=1$ to $N$}
			\State Randomly select a state family $\mathcal{F}$ from the predefined set.
			\If{$\mathcal{F}$ is an entangled family}
			\State Generate a pure state $\ket{\psi}$ from $\mathcal{F}$ with randomly sampled parameters (e.g., squeezing $r \in [0, 1.5]$, coherent amplitude $|\alpha| \in [0, 1]$).
			\State Sample mixing probability $p \sim U(0, 0.3)$.
			\State Create mixed state $\rho = (1-p)\ket{\psi}\bra{\psi} + p \frac{I}{d^M}$.
			\State \textbf{Verify Entanglement:} Compute negativity for all bipartite splits. If all are non-positive, discard the state and continue.
			\State Append valid $\rho$ to `states` and $1$ to `labels`.
			\ElsIf{$\mathcal{F}$ is a separable family}
			\State Generate $M$ single-mode states $\{\rho_j\}_{j=1}^M$ (random coherent states $\ket{\alpha}$ or Fock states $\ket{n}$ with $n < d$).
			\State Construct the product state $\rho = \bigotimes_{j=1}^M \rho_j$.
			\State Append $\rho$ to `states` and $0$ to `labels`.
			\EndIf
			\EndFor
			\State \textbf{Balance Dataset:} Undersample the majority class to ensure an equal number of separable and entangled samples.
			\State \textbf{Split Dataset:} Partition the balanced data into training (60\%), validation (20\%), and test (20\%) sets using a fixed random seed (42) for reproducibility.
			\State \textbf{Output:} Train, validation, and test sets of density matrices and corresponding labels.
		\end{algorithmic}
	\end{figure}
	
	\section{Statistical Analysis Methods}
	\label{app:stats}
	To ensure the robustness of our conclusions and the statistical significance of the observed performance gap, we employed a standard and rigorous statistical validation technique.
	
	\subsection{Stratified Bootstrap for Confidence Intervals}
	Point estimates of accuracy can be misleading, especially with a finite test set. To provide a reliable measure of the uncertainty in our performance metrics, we computed 95\% confidence intervals (CIs) using the **stratified bootstrap** method. The procedure is as follows:
	\begin{enumerate}
		\item Let the test set be $\mathcal{T} = \{(x_i, y_i)\}_{i=1}^{N_{\text{test}}}$.
		\item For $b=1, \dots, B$ (where we use $B=1000$ bootstrap resamples):
		\begin{enumerate}
			\item Create a new bootstrap sample $\mathcal{T}_b$ by drawing $N_{\text{test}}$ samples from $\mathcal{T}$ \emph{with replacement}.
			\item To maintain the class distribution of the original test set, we perform the sampling in a stratified manner: samples for the 'separable' class are drawn only from the separable subset of $\mathcal{T}$, and similarly for the 'entangled' class.
			\item Evaluate the accuracy of the trained model on the bootstrap sample $\mathcal{T}_b$, yielding an accuracy score $A_b$.
		\end{enumerate}
		\item The collection of scores $\{A_1, \dots, A_B\}$ forms the bootstrap distribution of the accuracy.
		\item The 95\% confidence interval is then calculated by taking the 2.5th and 97.5th percentiles of this distribution.
	\end{enumerate}
	An observed performance gap between two models is considered statistically significant if their 95\% CIs are non-overlapping. As seen in Table \ref{tab:3mode}, the CIs for the CV-QNN and the classical baselines are widely separated, providing strong evidence for the significance of the empirical gap.
	
\end{document}